\documentclass[a4paper,twoside]{article}

\usepackage{epsfig}
\usepackage{float}
\usepackage{subcaption}
\usepackage{calc}
\usepackage{amssymb}
\usepackage{amstext}
\usepackage{amsmath}
\usepackage{amsthm}
\usepackage{multicol}
\usepackage{pslatex}
\usepackage{apalike}
\usepackage{algorithm}
\usepackage{hyperref}
\usepackage{algorithmic}
\usepackage{booktabs}  
\usepackage{tabularx}  
\usepackage[bottom]{footmisc}
\usepackage{SCITEPRESS}     

\newtheorem{theorem}{Theorem}

\newtheorem{proposition}[theorem]{Proposition}

\newtheorem{definition}{Definition}

\begin{document}

\title{Reducing QUBO Density by Factoring Out Semi-Symmetries}

\author{Jonas Nüßlein\sup{1}, Leo Sünkel\sup{1}, Jonas Stein\sup{1}, Tobias Rohe\sup{1}, Daniëlle Schuman\sup{1}, Sebastian Feld\sup{2}, Corey O'Meara\sup{3}, Giorgio Cortiana\sup{3} and Claudia Linnhoff-Popien\sup{1}
\affiliation{\sup{1}Institute of Computer Science, LMU Munich, Germany}
\affiliation{\sup{2}Quantum \& Computer Engineering Department, Delft University of Technology, The Netherlands}
\affiliation{\sup{3}E.ON Digital Technology GmbH, Germany}
\email{jonas.nuesslein@ifi.lmu.de}
}


\abstract{
Quantum Approximate Optimization Algorithm (QAOA) and Quantum Annealing are prominent approaches for solving combinatorial optimization problems, such as those formulated as Quadratic Unconstrained Binary Optimization (QUBO). These algorithms aim to minimize the objective function $x^T Q x$, where $Q$ is a QUBO matrix. However, the number of two-qubit CNOT gates in QAOA circuits and the complexity of problem embeddings in Quantum Annealing scale linearly with the number of non-zero couplings in $Q$, contributing to significant computational and error-related challenges.  
To address this, we introduce the concept of \textit{semi-symmetries} in QUBO matrices and propose an algorithm for identifying and factoring these symmetries into ancilla qubits. \textit{Semi-symmetries} frequently arise in optimization problems such as \textit{Maximum Clique}, \textit{Hamilton Cycles}, \textit{Graph Coloring}, and \textit{Graph Isomorphism}. We theoretically demonstrate that the modified QUBO matrix $Q_{\text{mod}}$ retains the same energy spectrum as the original $Q$. Experimental evaluations on the aforementioned problems show that our algorithm reduces the number of couplings and QAOA circuit depth by up to $45\%$. For Quantum Annealing, these reductions also lead to sparser problem embeddings, shorter qubit chains and better performance. This work highlights the utility of exploiting QUBO matrix structure to optimize quantum algorithms, advancing their scalability and practical applicability to real-world combinatorial problems. 
}

\onecolumn \maketitle \normalsize \setcounter{footnote}{0} \vfill

\section{\uppercase{Introduction}}

The Quantum Approximate Optimization Algorithm (QAOA) \cite{farhi2014quantum} is designed to tackle combinatorial optimization problems using quantum computers by preparing a quantum state that maximizes the expectation value of the cost-hamiltonian. QAOA is widely recognized as a prime contender for showcasing quantum advantage on Noisy Intermediate-Scale Quantum (NISQ) devices \cite{zou2023multiscale}. It aims to approximate the ground state of a given physical system, often referred to as the Hamiltonian. However, its successful implementation faces challenges due to the high error rates inherent in current near-term quantum devices, which lack full error correction capabilities.

Utilizing QAOA to solve a problem entails a two-step process. Initially, the problem is translated into a parametric quantum circuit consisting of $p$ layers each consisting of $2$ adjustable parameters, where $p$ is a hyperparameter that needs to be set manually. This circuit is then run for thousands of trials. Subsequently, a classical optimizer utilizes the expectation value of the output distribution to refine the parameters. This iterative process continues until the optimal parameters for the circuit are determined. The cost function, which QAOA tries to minimize is usually represented as a Quadratic Unconstrained Binary Optimization (QUBO) problem.

The quantity of two-qubit CNOT operations within a QAOA circuit is equal to $2C \cdot p$ where $C$ is the number of non-zero couplings in the QUBO (number of edges in the problem graph). However, CNOT operations are susceptible to errors and often lead to prolonged runtimes. For instance, on the Google Sycamore quantum processor \cite{ayanzadeh2023frozenqubits}, CNOTs exhibit an average error-rate of $1\%$. Furthermore, CNOT gates might require additional SWAP gates since control- and target-qubit might not be connected on the hardware chip. Therefore, minimizing the number of these operations becomes crucial to improve the efficiency and accuracy of quantum optimization algorithms like QAOA.

In this paper, we therefore propose a method for using ancilla qubits to reduce the number of non-zero couplings and therefore also the number of CNOT operations and the depth of the QAOA circuit. We will also show that these sparser (but larger) QUBO matrices are easier to solve using Quantum Annealing since lower density leads to shorter physical qubit chains.

We present the concept of \textit{semi-symmetries} (see Definition 2) which we factor out into ancilla qubits. Our algorithm can therefore create different QUBO matrices that represent the same low-energy spectrum. To demonstrate the effectiveness of our approach, we tested it on four well-known optimization problems: Maximum Clique, Hamilton Cycles, Graph Coloring, and Graph Isomorphism. Our results show that our method can reduce the number of couplings and QAOA circuit depth by up to $45\%$, thus significantly improving the efficiency and scalability of Quantum Annealing and QAOA for solving a diverse range of NP-hard optimization problems.

\section{Background}

\subsection{Quadratic Unconstrained Binary Optimization}
Let $Q$ be a symmetric, real-valued $(n \times n)$-matrix and $x \in \mathbb{B}^n$ be a binary vector. Quadratic Unconstrained Binary Optimization (QUBO) ~\cite{zielinski2023pattern,roch2023effect} is an optimization problem of the form:
\begin{equation}
x^* = \underset{x}{argmin} \: H(x) = \underset{x}{argmin} \: \sum_{i=1}^{n}\sum_{j=i}^{n}{x_i\ x_j\ Q_{ij}}
\end{equation}

The function $H(x)$ is usually called \textit{Hamiltonian}. We will refer to the matrix $Q$ as the ``QUBO matrix''. The task is to find a binary vector $x$ that is as close as possible to the optimum which is known to be \textit{NP}-hard \cite{glover2018tutorial}. QUBOs attracted special attention recently since they can be solved using Quantum Optimization approaches like Quantum Annealing (QA) \cite{morita2008mathematical} or QAOA \cite{farhi2014quantum} which promises speed-ups compared to classical algorithms \cite{farhi2016quantum}. Numerous problems have already be encoded as a QUBO formulation \cite{bucher2023dynamic,zielinski2023influence,nusslein2023black}.

To solve a QUBO matrix using Quantum Annealing (QA), it must first be embedded onto a specialized graph ~\cite{4,10,12}. This process involves representing each logical qubit with multiple physical qubits. These physical qubits must be interconnected, forming what is known as a chain.

\subsection{QAOA}
The Quantum Approximate Optimization Algorithm (QAOA) is a hybrid quantum-classical algorithm proposed by Farhi et al. in 2014 \cite{farhi2014quantum} for solving combinatorial optimization problems. Let $C(x)$ be a cost function, where $x$ represents a binary string encoding a possible solution. The goal is to find the $x$ that minimizes $C(x)$. QAOA encodes this optimization problem into a quantum circuit, which can be parameterized by angles $\gamma$ and $\beta$. The quantum circuit prepares a quantum state $|\psi(\gamma, \beta)\rangle$ that represents a superposition of all possible solutions. The quantum circuit consists of alternating layers of two types of operators: the cost operator $U_C$ and the mixer operator $U_B$. The cost operator is responsible for encoding the cost function into the quantum state, while the mixer operator is responsible for exploring different solutions efficiently. The quantum state $|\psi(\gamma, \beta)\rangle$ prepared by the circuit is given by:
\[
|\psi(\gamma, \beta)\rangle = e^{-i\gamma_p U_B}e^{-i\beta_p U_C} \cdots e^{-i\gamma_1 U_B}e^{-i\beta_1 U_C}|+\rangle^{\otimes n}
\]
where $|+\rangle^{\otimes n}$ represents the initial state of $n$ qubits initialized to the superposition state, and $U_C$ and $U_B$ are the cost and mixer operators, respectively which are applied $p$ times. $p$ is a hyperparameter that needs to be manually specified. The parameters $\gamma$ and $\beta$ control the evolution of the quantum state.

The next step involves optimizing the parameters  $\gamma$ and $\beta$ to minimize the expectation value of the cost function. This optimization process is typically performed using classical optimization algorithms such as gradient descent or genetic algorithms. Given the quantum state $|\psi(\gamma, \beta)\rangle$, the expectation value of the cost function can be calculated as $E(\gamma, \beta) = \langle \psi(\gamma, \beta) | C | \psi(\gamma, \beta) \rangle
$. The goal is to find the optimal parameters $\gamma^*$ and $\beta^*$ that minimize $E(\gamma, \beta)$. This optimization process involves iteratively updating the parameters.

\subsection{Maximum Clique}
In graph theory, the Maximum Clique problem involves finding the largest subset of vertices $V' \subseteq V$ in a graph $G(V,E)$ such that every pair of vertices is connected by an edge. This problem has extensive applications across various domains, including social network analysis and bioinformatics~\cite{eblen2011maximum,rossi2015parallel}. To formulate the Maximum Clique problem as a QUBO problem, binary variables $x_i$ are used for each vertex $i$, where $x_i = 1$ indicates that vertex $i$ is included in the clique, and $x_i = 0$ otherwise. The Hamiltonian can therefore be written as:

\[
H(x) = \sum_i -x_i + A \cdot \sum_{(i,j) \in \overline{E}} x_i x_j 
\]

The second summand of $H$ enforces the solution to be a clique while the first summand rewards larger cliques \cite{lucas2014ising}.

\subsection{Hamilton Cycles}
Let $G(V,E)$ be a graph. The Hamilton Cycles problem asks if there is a path that starts from vertex $v_0$, visits every other vertice exactly once, and ends in vertex $v_0$ \cite{lucas2014ising}. This problem has practical applications in various fields, including logistics, transportation, and circuit design \cite{kawarabayashi2001survey,laporte2007locating}. To formulate this problem as a QUBO we introduce binary variables $x_{i,j}$ with $i \in [1..|V|]$ and $j \in [1..|V|]$. $x_{i,j} = 1$ iff vertex $i$ is at position $j$ of the cycle. The Hamiltonian can now be written as ~\cite{nusslein2022algorithmic}:

\begin{equation*}
\begin{aligned}
    H(x) &= \sum_i -x_i + A \cdot \sum_{i,j} \sum_{k,l} x_{i,j} x_{k,l} \cdot I[i=k \lor j=l \\ &\lor (l = j + 1 \land (i,k) \notin E) \lor (l = |V| - 1 \land j = 0 \\ &\land (i,k) \notin E)]
\end{aligned}
\end{equation*}

$H$ consists of three constraints: (1) each vertex must be visited (2) two vertices can't be at the same position in the cycle (3) two vertices can not be in neighboring positions of the cycle if there is no edge in the graph connecting them.

\subsection{Graph Coloring}
The Graph Coloring problem encompasses a wide range of applications from scheduling to register allocation in compilers, and even to radio frequency assignment in wireless communication networks \cite{ahmed2012applications}. At its core, the problem revolves around assigning colors to the vertices of a graph in such a way that no two adjacent vertices share the same color. Let $G=(V,E)$ be a graph, and $k$ be the number of available colors. To formulate this problem as a QUBO we introduce binary variables $x_{i,k}$ representing the assignment of color $k$ to vertex $i$ \cite{lucas2014ising}.

\begin{equation*}
\begin{aligned}
    H(x) &= \sum_{i,k} -x_{i,k} + A \cdot \sum_{i,k_1} \sum_{j,k_2} x_{i,k_1} x_{j,k_2} \cdot I[i=j \\ &\lor (k_1 = k_2 \land (i,j) \in E)]
\end{aligned}
\end{equation*}

$H$ encodes the two constraints that each vertex can only have one color and two adjacent vertices can't have the same color.

\subsection{Graph Isomorphism}
Graph Isomorphism (GI) is an important problem in graph theory that asks whether two graphs are structurally equivalent, albeit possibly differing in their vertex and edge labels. Formally, two graphs $G_1=(V_1, E_1)$ and $G_2=(V_2, E_2)$ are considered isomorphic if there exists a bijective mapping between their vertices such that their edge structures remain unchanged. In contrast to Maximum Clique, Hamilton Cycles and Graph Coloring, the complexity class for GI is still unknown (although it is expected to be in NP-intermediate) \cite{npintermediate}.

To formulate GI as a QUBO problem, we introduce binary variables $x_{i,j}$ representing the mapping of vertex $i$ of $G_1$ to vertex $j$ of $G_2$. The Hamiltonian can now be formulated as \cite{lucas2014ising}:

\begin{equation*}
\begin{aligned}
    H(x) &= \sum_i -x_i + A \cdot \sum_{i_1, j_1} \sum_{i_2, j_2} x_{i_1, j_1} x_{i_2, j_2} \cdot \\ &I[i_1 = j_2 \lor j_1 = j_2 \lor ((i_1, i_2) \in E_1 \land (j_1, j_2) \notin E_2) \\ &\lor ((i_1, i_2) \notin E_1 \land (j_1, j_2) \in E_2)]
\end{aligned}
\end{equation*}

\section{Related Work}
In this paper, we propose the concept of \textit{Semi-Symmetries} in QUBO matrices $Q$ and an algorithm for factoring them out into ancilla qubits to reduce the number of couplings and therefore the number of CNOT gates and circuit depth in QAOA and the chain length in QA. There are already two well-known types of symmetries in QUBO matrices: \textit{bit-flip-symmetry} and \textit{qubit-permutation-symmetry} \cite{shaydulin2021error,shaydulin2021exploiting,shaydulin2020classical}. Symmetry is defined here regarding the solution vectors $\{x\}$ and their associated energies $\{x^TQx\}$.

\subsection{Bit-flip-symmetry}

\textit{Bit-flip-symmetry} denotes the property of QUBOs that the inverse bit vector $x_I = 1 - x$ to a bit vector $x$ both have the same energy: ${(x_I)}^TQx_I = x^TQx$. \textit{Bit-flip-symmetries} occur, for example, in the Max-Cut problem:
\[
H(x) = \sum_{(i,j) \in E} - x_i - x_j + 2 x_i x_j
\]
\ \\
\textit{Bit-flip-symmetry} can be identified in a QUBO matrix $Q$ by substituting $x_i \gets (1 - x_i)$ and $x_j \gets (1 - x_j)$:
\begin{equation*}
\begin{aligned}
    H(x) &= \sum_{(i,j) \in E} - (1 - x_i) - (1 - x_j) + 2 (1 - x_i) (1 - x_j) = \\ &= \sum_{(i,j) \in E} -2 + x_i + x_j + 2 (1 - x_j - x_i + x_i x_j) = \\ &= \sum_{(i,j) \in E} - x_i - x_j + 2 x_i x_j
\end{aligned}
\end{equation*}
\ \\
Since the energy stays the same the QUBO contains a bit-flip-symmetry. Eliminating \textit{bit-flip-symmetry} can be done by removing the last qubit and assigning it the value $0$. Then, the remaining $(n-1) \times (n-1)$ QUBO matrix still encodes the original Hamiltonian.

\subsection{Qubit-permutation-symmetry}
Qubits $i$ and $j$ are \textit{qubit-permutation-symmetrical} if they have the same coupling values to all other qubits, i.e.:
\[
\forall \: k \in [1..n] : Q_{i,k} = Q_{j,k}
\]
This implies that for all $x^{(i=1, j=0)}$ it holds:
\[ H(x^{(i=1, j=0)}) = H(x^{(i=0, j=1)}) \] We use the notation $x^{(i=1, j=0)}$ for an arbitrary solution vector $x$ with qubit $i$ having value $1$ and qubit $j$ having value $0$. However, a trivial reduction of such a QUBO is not possible, since there are $3$ cases that have different energies: $x^{(i=0, j=0)}$, $x^{(i=1, j=0)}$ and $x^{(i=1, j=1)}$.

\subsection{Choosing a value for $p$}
Several works \cite{niu2019optimizing,pan2022automatic,ni2023more} have analyzed the influence of circuit depth on the performance of QAOA. Note that \textit{depth} is sometimes used synonymously with the number of layers, which we refer to as $p$. In this paper, we exclusively refer to \textit{depth} as the depth of the transpiled quantum circuit. To select the optimal number of repetitions \textit{p}, several approaches have been proposed for automatically setting this hyperparameter \cite{pan2022automatic,ni2023more,pan2022efficient,lee2021parameters}. In our experiments, we always used $p = 1$.

\subsection{Other approaches for eliminating couplings in $Q$}

In \textit{Algorithm 1}, the original $Q$ is modified by factoring out \textit{semi-symmetries} into additional ancilla qubits. However, we show that in doing so, the energy landscape for valid solutions is not altered. In contrast, there are heuristic approaches that alter the energy landscape to simplify $Q$. For example, in the paper \cite{sax2020approximate}, an approach was introduced to reduce the number of couplings in a QUBO by simply setting the smallest couplings to $0$ since they have the smallest influence on the energy landscape. By altering the energy landscape in this manner, it can no longer be guaranteed that the optimal solution $x_{mod}^*$ of the modified QUBO $Q_{mod}$ corresponds to the optimal solution $x^*$ of the original QUBO $Q$. 

Ising graphs associated to real-world problems, such as Airport Traffic Graphs, often exhibit a power-law structure \cite{ayanzadeh2023frozenqubits}, where some nodes have many more connections than others. In the paper \cite{ayanzadeh2023frozenqubits}, an approach is presented on how to partition the graphs with respect to these 'hubs'. This eliminates many couplings of the Hamiltonian, and the individual subgraphs can then be solved individually using a divide-and-conquer approach. A detailed analysis of the performance of QAOA depending on the graph structure is provided in \cite{herrman2021impact}. In \cite{ponce2023graph}, an approach is proposed on how large Max-Cut QUBOs can be solved by decomposing them into many smaller QUBOs. A similar approach is pursued in \cite{majumdar2021depth}.

There are already several papers \cite{shaydulin2021error,shaydulin2021exploiting,shaydulin2020classical} that exploit symmetries in QUBOs to generate more efficient and shorter QAOA circuits. In \cite{shaydulin2021error}, a method is proposed for leveraging \textit{bit-flip-symmetry} and \textit{qubit-permutation-symmetry} on Max-Cut graphs. In \cite{shaydulin2020classical} various types of symmetries that are relevant to QAOA and classical optimization problems are discussed. One prominent type is variable (qubit) permutation symmetries, which are transformations that rearrange the qubits of the quantum state without changing the problem's objective function. Such a symmetry can be caused when a graph contains automorphisms (a mapping of the graph to itself). The authors show that if a group of variable permutations leaves the objective function invariant, then the output probabilities of QAOA will be the same across all bit strings connected by such permutations, regardless of the chosen QAOA parameters and depth which can be used to reduce the dimension of the effective Hilbert space.

\section{Algorithm}

We start this section by providing a formal definition of \textit{conflicting qubits} and \textit{semi-symmetries}.

\begin{definition}[Conflicting qubits]
Let $H(x) = x^T Q x$ be the energy of a solution $x$. Qubits $i$ and $j$ are called \textit{conflicting}, iff for every solution $x^{(i=1, j=1)}$ it holds that: \[ H(x^{(i=1, j=1)}) > \{ H(x^{(i=1, j=0)}), H(x^{(i=0, j=1)}), H(x^{(i=0, j=0)}) \} \].
\end{definition}

\begin{definition}[Semi-symmetry]
Conflicting qubits $(i,j)$ are semi-symmetric if and only if:
\[
\exists \: U \subseteq \{1..n\} \backslash \{i,j\} \land |U| \geq 3 : \forall \: k \in U : Q_{i,k} = Q_{j,k} \neq 0
\]
In other words, two \textit{conflicting qubits} $(i,j)$ are \textit{semi-symmetric} iff there are at least $3$ other qubits to which $i$ and $j$ have the same non-zero couplings. This is a weakened definition of symmetry compared to \textit{qubit-permutation-symmetry} where qubits $i$ and $j$ needed the same couplings to \textit{all} other qubits.
\end{definition}

\begin{algorithm}
   \caption{Factoring Semi-Symmetries}
   \label{alg:factoring_syms}
\begin{algorithmic}
   \STATE \textbf{Input:} QUBO matrix $Q$ of size $n \times n$
   \STATE \hspace{13mm} number of ancillas $numAncillas \in \mathbb{N}$
   \STATE \hspace{13mm} parameter $z \in \mathbb{R}^+$
   \STATE
   \STATE $n_{new} = n$ 
   \STATE $cL = \textsc{GetConflictList}(Q, n_{new})$
   \STATE
   \WHILE{$\text{len}(cL) > 0$}
       \STATE $syms, (i,j) = \textsc{GetMostSymQubits}(Q, n_{new}, cL)$
       \IF{$\text{len}(syms) < 3 \textbf{ or } n_{new} = n + numAncillas$}
           \STATE \textbf{break}
       \ENDIF
       \STATE $n_{new} = n_{new} + 1$
       \STATE $Q = \textsc{Enhance}(Q, n_{new}, (i,j), syms)$
       \STATE $cL = \textsc{GetConflictList}(Q, n_{new})$
   \ENDWHILE
   \STATE
   \STATE \textbf{return} $Q$
   \STATE

   \STATE \textbf{function} \textsc{GetConflictList}($Q, n$)
       \STATE $cL = []$
       \STATE $Z = [ \sum_{j \in [1..n], Q_{i,j} < 0} Q_{i,j} \; : \; i \in [1..n] ]$
       \FOR{$i = 1$ \textbf{to} $n$, $j = 1$ \textbf{to} $n$}
           \IF{$i < j \ \textbf{and} \ Q_{i,j} > -Z[i] - Z[j]$}
               \STATE $cL.\text{append}((i,j))$
           \ENDIF
       \ENDFOR
       \STATE \textbf{return} $cL$
   \STATE \textbf{end function}
   \STATE

   \STATE \textbf{function} \textsc{GetMostSymQubits}($Q, n, cL$)
       \STATE $best = (0, 1)$
       \STATE $bestSyms = []$
       \FOR{$(i,j) \in cL$}
           \STATE $syms = [ k \in [1..n] \; : \; Q_{i,k} = Q_{j,k} \neq 0 ]$
           \IF{$\text{len}(syms) \geq \text{len}(bestSyms)$}
               \STATE $best = (i,j)$
               \STATE $bestSyms = syms$
           \ENDIF
       \ENDFOR
       \STATE \textbf{return} $bestSyms, best$
   \STATE \textbf{end function}
   \STATE

   \STATE \textbf{function} \textsc{Enhance}($Q, n, (i,j), syms$)
       \STATE $Q_{i,i} = Q_{i,i} + z$
       \STATE $Q_{j,j} = Q_{j,j} + z$
       \STATE $Q_{n,n} = z$
       \STATE $Q_{i,n} = -2 \cdot z$
       \STATE $Q_{j,n} = -2 \cdot z$
       \STATE $Q_{i,j} = 2 \cdot z$
       \FOR{$k \in syms$}
           \STATE $Q_{k,n} = Q_{i,k}$
           \STATE $Q_{i,k} = 0$
           \STATE $Q_{j,k} = 0$
       \ENDFOR
       \STATE \textbf{return} $Q$
   \STATE \textbf{end function}
\end{algorithmic}
\end{algorithm}

\subsection{Proof-of-Concept Example}
In the following section, we demonstrate our algorithm for a simple proof-of-concept example. To do this, we consider the following graph:

\begin{figure}[H]
 \centering
  \includegraphics[scale=0.27]{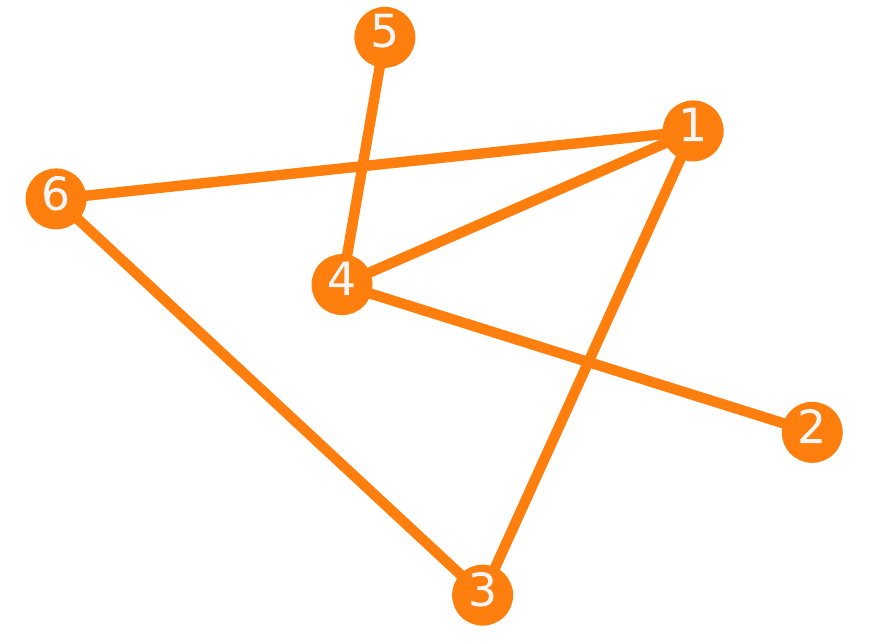}
  \caption{A simple proof-of-concept graph}\label{fig:main_charts}
\end{figure}

We now want to find the largest clique (Maximum Clique) for the graph $G=(V,E)$ in Figure 1, i.e. the largest set of nodes for which each pair of nodes is connected by an edge. The Hamiltonian that encodes this problem is given by :

\[
H = \sum_i -x_i + \sum_{(i,j) \in \overline{E}} 3 x_i x_j
\]

The QUBO matrix $Q$ for Maximum Clique and the graph from \textbf{Figure 1} is listed in \textbf{Table I (upper)}. It requires $6$ qubits and $9$ couplings. $Q$ contains a \textit{semi-symmetry} between qubits $2$ and $5$ which can be factored out into an additional ancilla qubit $7$ (see \textbf{Table I (lower)}). The modified QUBO matrix $Q_{mod}$ requires $7$ qubits but only $8$ couplings.

\begin{figure*}[t!]
\centering
\minipage{0.99\textwidth}
  \centering
  \includegraphics[width=\linewidth]{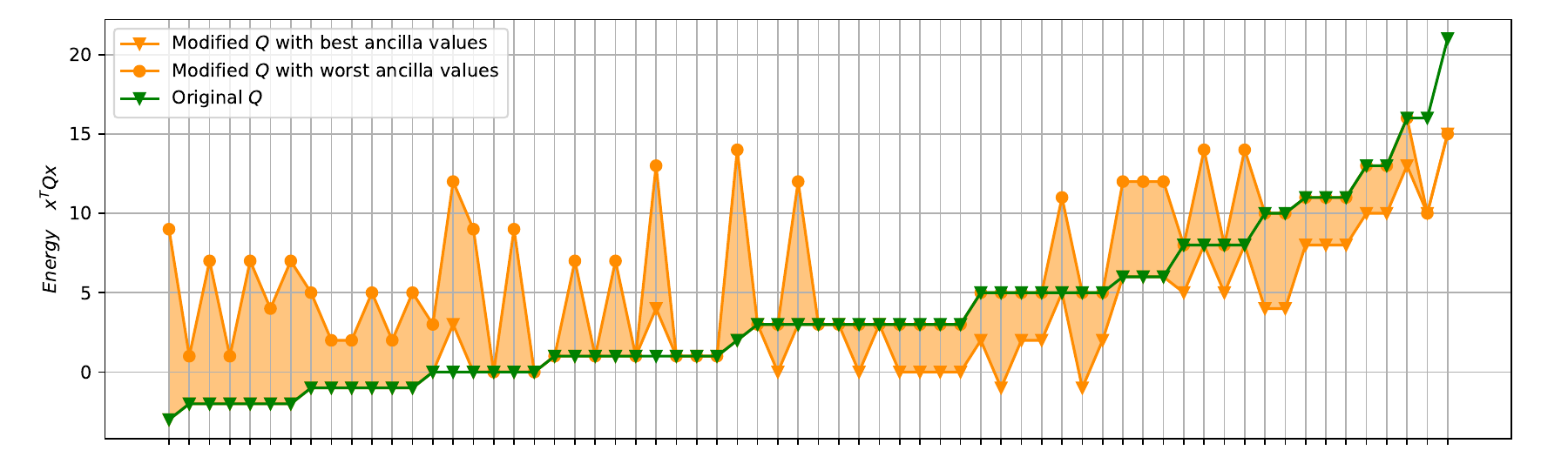}
  \includegraphics[width=\linewidth]{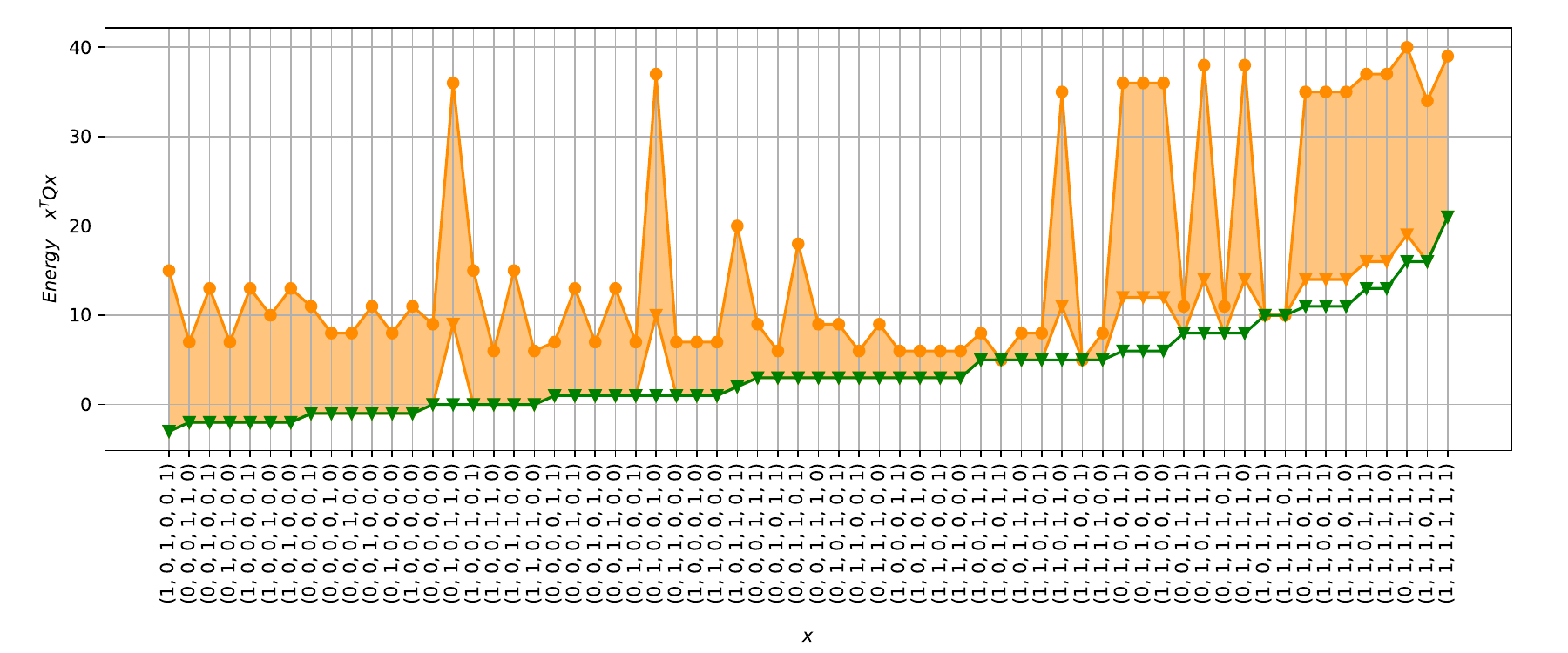}
  \caption{The green line represents the sorted energy spectrum of the left QUBO in Table I. The orange lines represent the energy spectrum of the right QUBO in Table I with the lower graph in both plots being the energetically more favorable choice of the ancilla value while the upper graph in both plots representing the energetically less favorable choice. The upper plot represents the energy spectrum of $Q_{mod}$ with $z=3$ and the lower plot with $z=9$. For $z=3$ we can see that invalid solutions can have a lower energy in $Q_{mod}$ than in $Q$ but if we increase $z$ all invalid solutions have an energy equal or higher than in $Q$. But even $z=3$ is already sufficient for the global optimum $x^*$ in $Q$ to also be the global optimum in $Q_{mod}$.}\label{fig:main_charts}
\endminipage
\end{figure*}

\begin{table}[H]
  \centering
    \begin{tabular}{|l|l|l|l|l|l|} \hline 
         -1& 3& & & 3& \\\hline
         & -1& 3& & 3& 3\\\hline
         & & -1& 3& 3& \\\hline
         & & & -1& & 3\\\hline
         & & & & -1& 3\\\hline
         & & & & & -1\\\hline
    \end{tabular}
    \vspace{1em}
    \vspace{1em}
    \begin{tabular}{|l|l|l|l|l|l|l|} \hline 
         -1& & & & & & 3\\\hline
         & 2& & & 9& & -6\\\hline
         & & -1& 3& & & 3\\\hline
         & & & -1& & 3& \\\hline
         & & & & 2& & -6\\\hline
         & & & & & -1& 3\\\hline
         & & & & & & 3\\\hline
    \end{tabular}
    \caption{(Upper) QUBO matrix $Q$ for Maximum Clique and the graph in Figure 1. (Lower) Modified QUBO matrix $Q_{mod}$ of $Q$ using \textbf{Algorithm 1}. The \textit{semi-symmetry} between qubits $2$ and $5$ was factored out into an additional ancilla qubit.}
\end{table}

In the following section, we theoretically show that our algorithm doesn't change the energy landscape for valid solutions and in section \textit{4.3} we analyze the energy spectra for both QUBO matrices in \textit{Table I}.

\subsection{Theoretical Analysis for Correctness}
We can prove that our modified QUBO $Q_{mod}$ has the same optimal solutions as $Q$ with the best choice of ancilla values, i.e. \textbf{Algorithm 1} doesn't change the energies of valid solutions and doesn't decrease the energy of invalid solutions. Valid solutions $x$ are bit-vectors that don't violate \textit{conflicting qubit} constraints, i.e. if $(i,j)$ are conflicting then $x_i = 0$ or $x_j = 0$. Invalid solutions are bit-vectors with $x_i = 1$ and $x_j = 1$. 

\begin{proposition}
If we choose $z = \sum_{(i,j)} |Q_{i,j}|$, valid solutions $x$ have the same energy regarding $Q$ as to $Q_{mod}$ with the best values for the ancilla qubits $x_{mod} = x + [x_a]$. The energy of invalid solution doesn't decrease with respect to $Q_{mod}$ even with the best ancilla values.
\end{proposition}
\begin{proof}
Let $Q$ be any QUBO matrix and $x \in \mathbb{B}^n$ be any solution vector. The energy $E$ for $x$ corresponds to $E = x^T Q x$. Let $(x_i,x_j)$ be a pair of conflicting qubits, i.e. no valid solution $x$ contains assignments $i=1$ and $j=1$ at the same time. Further assume that $(x_i,x_j)$ are semi-symmetrical and \textbf{Algorithm 1} factored out the semi-symmetries into an ancilla qubit $x_a$.

\ \\
\textbf{Case 1: $x_i = 0, x_j = 0, x_a = 0$}: in this case, we can easily see that the energy of $x_{mod} = x + [0]$ regarding $Q_{mod}$ is identical to the original energy: $E_{mod} = {(x + [0])}^T \cdot Q_{mod} \cdot (x + [0]) = E$.

\ \\
\textbf{Case 2: $x_i = 0, x_j = 0, x_a = 1$}: in this case the modified energy corresponds to: $E_{mod} = E + z + \sum_{k \in syms} Q_{i,k}$. Since we can choose $z = \sum_{(i,j)} |Q_{i,j}|$, it holds that $z + \sum_{k \in syms} Q_{i,k} \geq 0$. Therefore: $E_{mod} \geq E$.

\ \\
\textbf{Case 3: $x_i = 1, x_j = 0, x_a = 0$}: $E_{mod} = E + z - \sum_{k \in syms} Q_{i,k}$. Again, since $z = \sum_{(i,j)} |Q_{i,j}|$, it holds that $z - \sum_{k \in syms} Q_{i,k} \geq 0$. Therefore: $E_{mod} \geq E$.

\ \\
\textbf{Case 4: $x_i = 1, x_j = 0, x_a = 1$}: $E_{mod} = E + z + z - 2z + \sum_{k \in syms} Q_{i,k} - \sum_{k \in syms} Q_{i,k} = E$.

\ \\
\textbf{Case 5: $x_i = 0, x_j = 1, x_a = 0$}: analogous to case 3.

\ \\
\textbf{Case 6: $x_i = 0, x_j = 1, x_a = 0$}: analogous to case 4.

\ \\
\textbf{Case 7: $x_i = 1, x_j = 1, x_a = 0$}: $E_{mod} = E + z + z + 2z - \sum_{k \in syms} Q_{i,k} - \sum_{k \in syms} Q_{i,k} > E$.

\ \\
\textbf{Case 8: $x_i = 1, x_j = 1, x_a = 1$}: $E_{mod} = E + z + z + z - 2z - 2z + 2z - \sum_{k \in syms} Q_{i,k} - \sum_{k \in syms} Q_{i,k} + \sum_{k \in syms} Q_{i,k} = E + z - \sum_{k \in syms} Q_{i,k} \geq E$.

\ \\
The best choices for the ancilla qubit for valid solutions are \textit{case 1}, \textit{case 4} and \textit{case 6} which all have energy $E$. Therefore, the energy did not change for valid solutions. For invalid solutions (\textit{cases 7} and \textit{8}) the energy does not decrease.

\end{proof}

\subsection{Empirical Evaluation of the Energy Landscape for the PoC}
We now empirically investigate this theoretical finding in our proof-of-concept example. Since there are $6$ qubits, there are $2^6=64$ possible solutions $x$. For each $x$ we calculated the energy regarding $Q$ (Table I, left), see green lines in Figure 2. Further, we have calculated the energy in the modified QUBO (Table I, right) with both possible values (0 and 1) for the ancilla qubit 7. Then we have plotted for each $x$ the original energy, the energy in the modified $Q_{mod}$ with the worse choice for the ancilla qubit and the better choice for the ancilla qubit.

The upper plot in Figure 2 shows the result with $z=3$, and the lower plot shows the result with $z=9$. We can verify the proposition if we choose $z$ big enough, but often a lower value for $z$ is already enough for the original optimal solution $x$ to also be the optimal solution in $Q_{mod}$.

\section{Experiments}

\begin{figure*}[t!]
\centering
\minipage{1\textwidth}
  \centering
  \includegraphics[width=\linewidth]{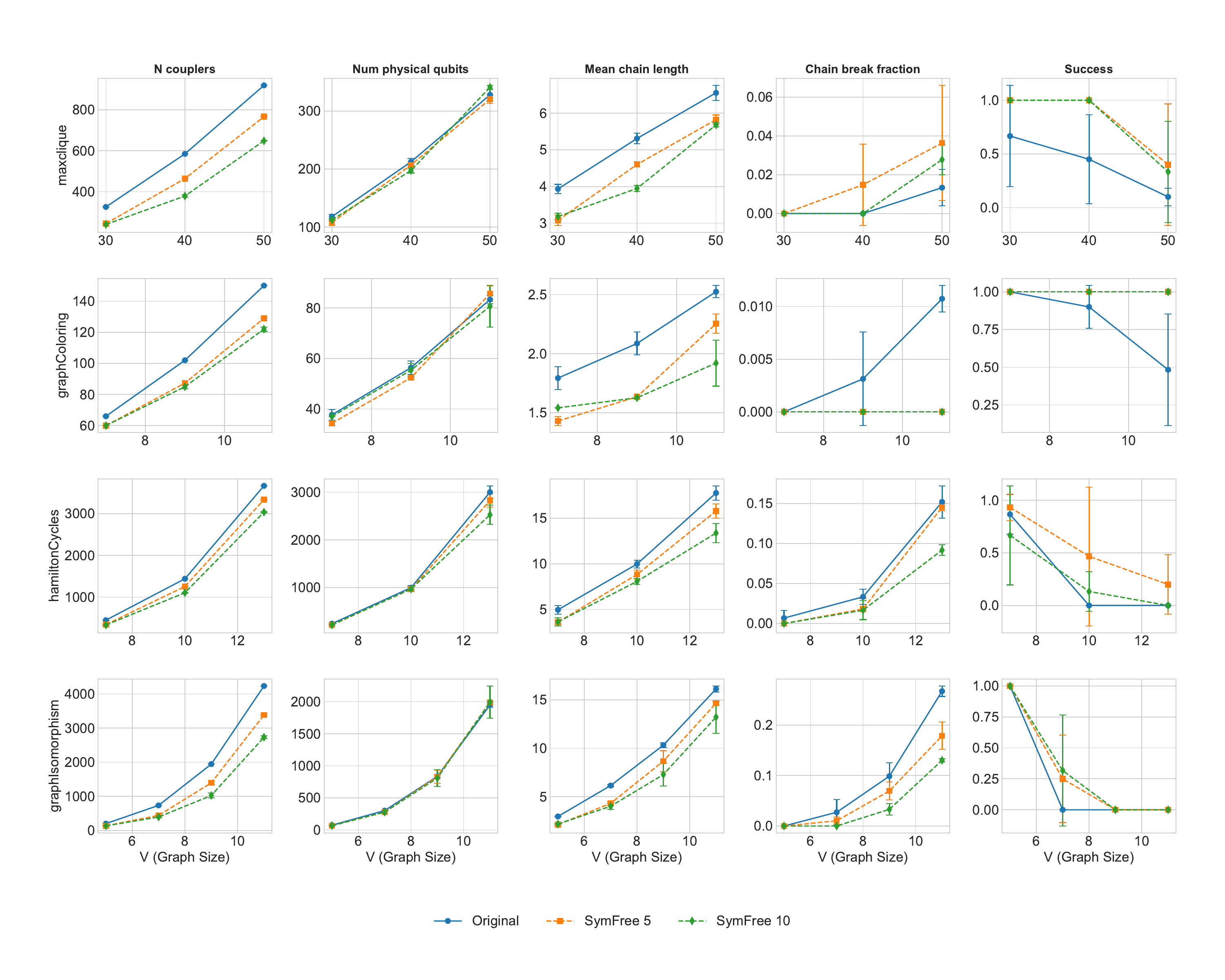}
  \caption{This figure presents the outcomes of our experiments on removing \textit{semi-symmetries} using the algorithm described in Section 3. We evaluated our approach on four optimization problems: Maximum Clique, Graph Coloring, Hamilton Cycles, and Graph Isomorphism. Along the horizontal axis, the plot shows increasing problem size. We considered five metrics: the number of couplings, the number of physical qubits, the mean chain length, the chain break fraction, and the success rate. The success metric equals 1 whenever the annealer reaches the global optimum, as verified by a classical heuristic solver. Our results indicate that removing semi-symmetries reduces the number of couplings, which in turn lowers the mean chain length. This reduction leads to fewer chain breaks and ultimately improves the success rate.}\label{fig:main_charts}
\endminipage
\end{figure*}

We conducted experiments on four representative combinatorial optimization problems — Maximum Clique, Graph Coloring, Hamilton Cycles, and Graph Isomorphism — to evaluate the impact of removing semi-symmetries from QUBO formulations. Each problem instance is characterized by the number of vertices $|V|$ and edges $|E|$, with additional parameters such as the number of colors $k$ for Graph Coloring. Our goal was to assess how introducing additional ancilla qubits (to remove semi-symmetries) influences key hardware-related metrics after embedding the resulting QUBOs onto a quantum annealer.

Figure 3 summarizes our findings. The figure consists of a $4 \times 5$ grid of plots, where each row corresponds to one of the four problems (from top to bottom: Maximum Clique, Graph Coloring, Hamilton Cycles, Graph Isomorphism) and each column represents a different metric. The metrics we considered were: the number of couplers used in the embedded QUBO, the total number of physical qubits, the mean chain length, the chain break fraction, and the probability of a successful solution (Success). The horizontal axis in every subplot denotes the number of vertices $|V|$ for the given problem instances, thus capturing how problem scale affects these metrics. We compared three scenarios for each problem and setting: the original QUBO (blue), and two symmetry-free variants obtained by introducing 5 (orange) or 10 (green) ancilla qubits to remove semi-symmetries.

As the problem size $|V|$ increased, the original QUBO instances tended to produce larger and more complex embeddings, reflected by a higher number of couplers, more physical qubits, and longer chains. These embedding characteristics often led to a higher chain break fraction and, consequently, a lower probability of success. In contrast, when we introduced ancilla qubits to remove semi-symmetries, both the 5- and 10-ancilla configurations showed a noticeable reduction in complexity: we observed fewer couplers, a smaller chain length, and a reduced chain break fraction. This improved embedding quality often translated into a higher success probability for finding the optimal solution, despite the growing complexity of the underlying problem instances.

\begin{figure*}[t]
\centering
\minipage{0.99\textwidth}
  \centering
  \includegraphics[width=\linewidth]{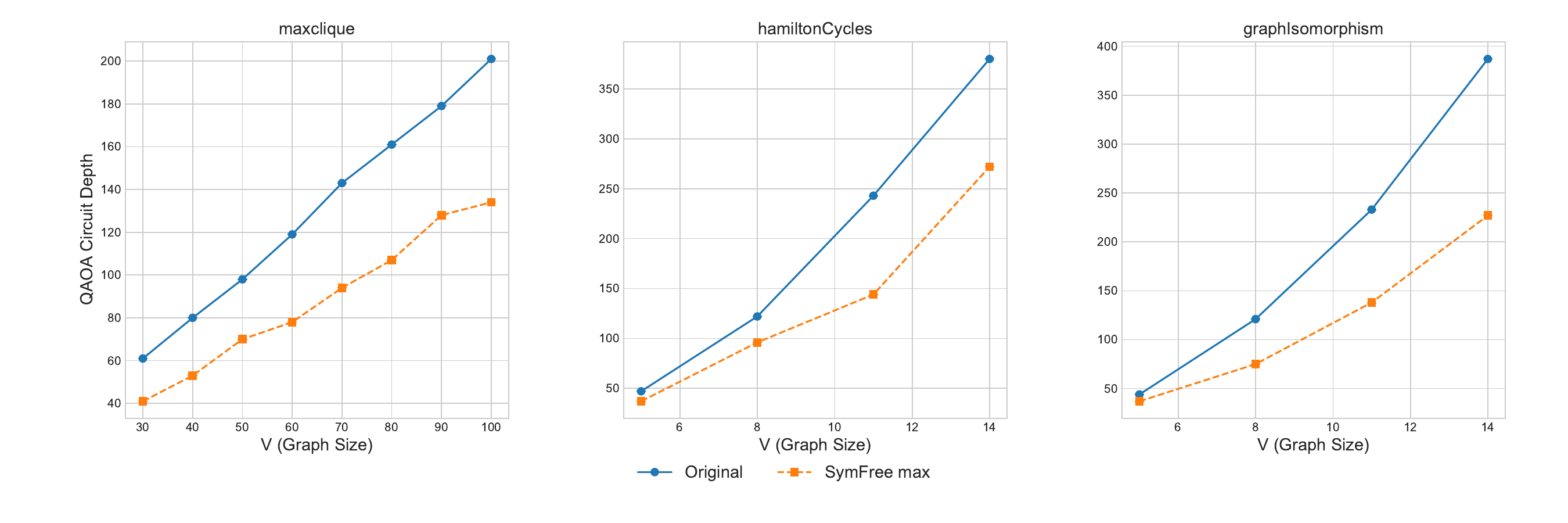}
  \caption{Comparison of QAOA circuit depth with and without semi-symmetry elimination for different graph problems: MaxClique (left), HamiltonCycles (center), and GraphIsomorphism (right). The blue solid line (\textit{Original}) represents the circuit depth without symmetry elimination, while the orange dashed line (\textit{SymFree max}) shows the results after factoring out all semi-symmetries. The reduction in circuit depth becomes more pronounced for larger graph sizes, demonstrating the efficiency of the semi-symmetry elimination approach.}\label{fig:main_charts}
\endminipage
\end{figure*}

Notably, the advantage of the symmetry-free approach became more pronounced at larger problem sizes.

In summary, these experiments demonstrate that leveraging semi-symmetry removal to produce symmetry-free QUBOs can yield more hardware-friendly embeddings. The resulting QUBOs generally require fewer resources (couplers, physical qubits) and produce higher-quality embeddings (shorter chains, lower chain break fractions), ultimately improving the probability of success. By making problem instances more tractable for current quantum annealers, our approach provides a practical path toward better performance on larger and more challenging optimization problems.
\ \\
\ \\
In a second experiment (see Figure 4), we analyzed the transpiled circuit depth of QAOA with $p = 1$. Unlike the first experiment, we factored out all semi-symmetries. In the first experiment, we used a fixed number of semi-symmetries, as eliminating smaller semi-symmetries (e.g. fewer than 10 common couplings) typically does not improve performance on quantum annealers. This is because the additional physical qubits required to represent the new logical (ancilla) qubits often exceed the physical qubits saved by shortening the chains. 

The results demonstrate that the real benefit of our approach becomes evident with larger problem sizes, as these inherently contain more semi-symmetries. As quantum hardware continues to advance, enabling the solution of larger problem instances, this method will become increasingly significant. The ability to effectively handle semi-symmetries at scale promises substantial improvements in circuit depth and resource efficiency, highlighting the long-term relevance of this technique beyond the currently solvable problem instances.

For the empirical evaluation in the first experiment, we used the D-Wave Advantage 4.1 quantum annealer, which features 5760 qubits. Each experimental configuration was executed over 10 runs, and the standard deviation is reported in the results to illustrate variability.

For the QAOA experiments, we transpiled the circuits using Qiskit Aer. However, the transpilation was performed solely to measure circuit depth, as the resulting circuits are too large to be executed or simulated on available quantum hardware.

\section{Conclusion}

In this work, we introduced the concept of semi-symmetries in QUBO matrices and proposed an algorithm for identifying and factoring these symmetries into ancilla qubits. Our method effectively reduces the number of non-zero couplings in the QUBO matrix, which directly translates to improvements in both the Quantum Approximate Optimization Algorithm (QAOA) and Quantum Annealing.

Theoretical analysis confirmed that the modified QUBO matrix $Q_{\text{mod}}$ retains the same energy spectrum as the original matrix $Q$, ensuring the correctness of the optimization problem. Our experimental evaluations demonstrated significant reductions in both computational and physical resource requirements. Specifically, our approach achieved up to a 45\% reduction in couplings and QAOA circuit depth. For Quantum Annealing, the reduced matrix structure led to sparser problem embeddings, shorter qubit chains, and improved overall performance.

The results were validated across a range of combinatorial optimization problems, including Maximum Clique, Hamilton Cycles, Graph Coloring, and Graph Isomorphism, all of which naturally exhibit semi-symmetries. The findings indicate that leveraging such symmetries enhances the scalability and efficiency of quantum optimization algorithms, addressing key challenges such as circuit depth, error accumulation, and embedding complexity.

Looking forward, as quantum hardware continues to advance, the ability to exploit matrix structure for optimization problems will become increasingly crucial. Our method provides a promising step toward making quantum algorithms more practical and scalable for real-world combinatorial problems. Future work will explore further generalizations of semi-symmetry detection, integration with higher-layer optimization frameworks, and broader applicability to other classes of quantum algorithms.

\section*{Acknowledgment}
This publication was created as part of the Q-Grid project (13N16179) under the ``quantum technologies -- from basic research to market'' funding program, supported by the German Federal Ministry of Education and Research.

\bibliographystyle{apalike}
{\small
\bibliography{example}}

\begin{thebibliography}{}

\bibitem[Ahmed, 2012]{ahmed2012applications}
Ahmed, S. (2012).
\newblock Applications of graph coloring in modern computer science.
\newblock {\em International Journal of Computer and Information Technology}, 3(2):1--7.

\bibitem[Ayanzadeh et~al., 2023]{ayanzadeh2023frozenqubits}
Ayanzadeh, R., Alavisamani, N., Das, P., and Qureshi, M. (2023).
\newblock Frozenqubits: Boosting fidelity of qaoa by skipping hotspot nodes.
\newblock In {\em Proceedings of the 28th ACM International Conference on Architectural Support for Programming Languages and Operating Systems, Volume 2}, pages 311--324.

\bibitem[Bucher et~al., 2023]{bucher2023dynamic}
Bucher, D., N{\"u}{\ss}lein, J., O'Meara, C., Angelov, I., Wimmer, B., Ghosh, K., Cortiana, G., and Linnhoff-Popien, C. (2023).
\newblock Dynamic price incentivization for carbon emission reduction using quantum optimization.
\newblock {\em arXiv preprint arXiv:2309.05502}.

\bibitem[Eblen et~al., 2011]{eblen2011maximum}
Eblen, J.~D., Phillips, C.~A., Rogers, G.~L., and Langston, M.~A. (2011).
\newblock The maximum clique enumeration problem: Algorithms, applications and implementations.
\newblock In {\em International Symposium on Bioinformatics Research and Applications}, pages 306--319. Springer.

\bibitem[Farhi et~al., 2014]{farhi2014quantum}
Farhi, E., Goldstone, J., and Gutmann, S. (2014).
\newblock A quantum approximate optimization algorithm.
\newblock {\em arXiv preprint arXiv:1411.4028}.

\bibitem[Farhi and Harrow, 2016]{farhi2016quantum}
Farhi, E. and Harrow, A.~W. (2016).
\newblock Quantum supremacy through the quantum approximate optimization algorithm.
\newblock {\em arXiv preprint arXiv:1602.07674}.

\bibitem[Glover et~al., 2018]{glover2018tutorial}
Glover, F., Kochenberger, G., and Du, Y. (2018).
\newblock A tutorial on formulating and using qubo models.
\newblock {\em arXiv preprint arXiv:1811.11538}.

\bibitem[Herrman et~al., 2021]{herrman2021impact}
Herrman, R., Treffert, L., Ostrowski, J., Lotshaw, P.~C., Humble, T.~S., and Siopsis, G. (2021).
\newblock Impact of graph structures for qaoa on maxcut.
\newblock {\em Quantum Information Processing}, 20(9):289.

\bibitem[Kawarabayashi, 2001]{kawarabayashi2001survey}
Kawarabayashi, K.-i. (2001).
\newblock A survey on hamiltonian cycles.
\newblock {\em Interdisciplinary Information Sciences}, 7(1):25--39.

\bibitem[Laporte and Mart{\'\i}n, 2007]{laporte2007locating}
Laporte, G. and Mart{\'\i}n, I.~R. (2007).
\newblock Locating a cycle in a transportation or a telecommunications network.
\newblock {\em Networks: An International Journal}, 50(1):92--108.

\bibitem[Lee et~al., 2021]{lee2021parameters}
Lee, X., Saito, Y., Cai, D., and Asai, N. (2021).
\newblock Parameters fixing strategy for quantum approximate optimization algorithm.
\newblock In {\em 2021 IEEE international conference on quantum computing and engineering (QCE)}, pages 10--16. IEEE.

\bibitem[Lodewijks, 2020]{12}
Lodewijks, B. (2020).
\newblock Mapping {NP}-hard and {NP}-complete optimisation problems to quadratic unconstrained binary optimisation problems.

\bibitem[Lu, ]{npintermediate}
Lu, X.
\newblock On np-intermediate, isomorphism problems, and polynomial hierarchy.

\bibitem[Lucas, 2014]{lucas2014ising}
Lucas, A. (2014).
\newblock Ising formulations of many np problems.
\newblock {\em Frontiers in physics}, 2:5.

\bibitem[Majumdar et~al., 2021]{majumdar2021depth}
Majumdar, R., Bhoumik, D., Madan, D., Vinayagamurthy, D., Raghunathan, S., and Sur-Kolay, S. (2021).
\newblock Depth optimized ansatz circuit in qaoa for max-cut.
\newblock {\em arXiv preprint arXiv:2110.04637}.

\bibitem[Mooney et~al., 2019]{10}
Mooney, G., Tonetto, S., Hill, C., and Hollenberg, L. (2019).
\newblock Mapping {NP}-hard problems to restructed adiabatic quantum architectures.

\bibitem[Morita and Nishimori, 2008]{morita2008mathematical}
Morita, S. and Nishimori, H. (2008).
\newblock Mathematical foundation of quantum annealing.
\newblock {\em Journal of Mathematical Physics}, 49(12).

\bibitem[Ni et~al., 2023]{ni2023more}
Ni, X.-H., Cai, B.-B., Liu, H.-L., Qin, S.-J., Gao, F., and Wen, Q.-Y. (2023).
\newblock More efficient parameter initialization strategy in qaoa for maxcut.
\newblock {\em arXiv preprint arXiv:2306.06986}.

\bibitem[Niu et~al., 2019]{niu2019optimizing}
Niu, M.~Y., Lu, S., and Chuang, I.~L. (2019).
\newblock Optimizing qaoa: Success probability and runtime dependence on circuit depth.
\newblock {\em arXiv preprint arXiv:1905.12134}.

\bibitem[N{\"u}{\ss}lein et~al., 2022]{nusslein2022algorithmic}
N{\"u}{\ss}lein, J., Gabor, T., Linnhoff-Popien, C., and Feld, S. (2022).
\newblock Algorithmic qubo formulations for k-sat and hamiltonian cycles.
\newblock In {\em Proceedings of the genetic and evolutionary computation conference companion}, pages 2240--2246.

\bibitem[N{\"u}{\ss}lein et~al., 2023]{nusslein2023black}
N{\"u}{\ss}lein, J., Roch, C., Gabor, T., Stein, J., Linnhoff-Popien, C., and Feld, S. (2023).
\newblock Black box optimization using qubo and the cross entropy method.
\newblock In {\em International Conference on Computational Science}, pages 48--55. Springer.

\bibitem[Pan et~al., 2022a]{pan2022efficient}
Pan, Y., Tong, Y., Xue, S., and Zhang, G. (2022a).
\newblock Efficient depth selection for the implementation of noisy quantum approximate optimization algorithm.
\newblock {\em Journal of the Franklin Institute}, 359(18):11273--11287.

\bibitem[Pan et~al., 2022b]{pan2022automatic}
Pan, Y., Tong, Y., and Yang, Y. (2022b).
\newblock Automatic depth optimization for a quantum approximate optimization algorithm.
\newblock {\em Physical Review A}, 105(3):032433.

\bibitem[Ponce et~al., 2023]{ponce2023graph}
Ponce, M., Herrman, R., Lotshaw, P.~C., Powers, S., Siopsis, G., Humble, T., and Ostrowski, J. (2023).
\newblock Graph decomposition techniques for solving combinatorial optimization problems with variational quantum algorithms.
\newblock {\em arXiv preprint arXiv:2306.00494}.

\bibitem[Prasanna et~al., 2019]{4}
Prasanna, D., Patton, R., Schuman, C., and Potok, T. (2019).
\newblock Efficiently embedding {QUBO} problems on adiabatic quantum computers.

\bibitem[Roch et~al., 2023]{roch2023effect}
Roch, C., Ratke, D., N{\"u}{\ss}lein, J., Gabor, T., and Feld, S. (2023).
\newblock The effect of penalty factors of constrained hamiltonians on the eigenspectrum in quantum annealing.
\newblock {\em ACM Transactions on Quantum Computing}, 4(2):1--18.

\bibitem[Rossi et~al., 2015]{rossi2015parallel}
Rossi, R.~A., Gleich, D.~F., and Gebremedhin, A.~H. (2015).
\newblock Parallel maximum clique algorithms with applications to network analysis.
\newblock {\em SIAM Journal on Scientific Computing}, 37(5):C589--C616.

\bibitem[Sax et~al., 2020]{sax2020approximate}
Sax, I., Feld, S., Zielinski, S., Gabor, T., Linnhoff-Popien, C., and Mauerer, W. (2020).
\newblock Approximate approximation on a quantum annealer.
\newblock In {\em Proceedings of the 17th ACM International Conference on Computing Frontiers}, pages 108--117.

\bibitem[Shaydulin and Galda, 2021]{shaydulin2021error}
Shaydulin, R. and Galda, A. (2021).
\newblock Error mitigation for deep quantum optimization circuits by leveraging problem symmetries.
\newblock In {\em 2021 IEEE International Conference on Quantum Computing and Engineering (QCE)}, pages 291--300. IEEE.

\bibitem[Shaydulin et~al., 2020]{shaydulin2020classical}
Shaydulin, R., Hadfield, S., Hogg, T., and Safro, I. (2020).
\newblock Classical symmetries and qaoa.
\newblock {\em arXiv preprint arXiv:2012.04713}.

\bibitem[Shaydulin and Wild, 2021]{shaydulin2021exploiting}
Shaydulin, R. and Wild, S.~M. (2021).
\newblock Exploiting symmetry reduces the cost of training qaoa.
\newblock {\em IEEE Transactions on Quantum Engineering}, 2:1--9.

\bibitem[Zielinski et~al., 2023a]{zielinski2023influence}
Zielinski, S., N{\"u}{\ss}lein, J., Stein, J., Gabor, T., Linnhoff-Popien, C., and Feld, S. (2023a).
\newblock Influence of different 3sat-to-qubo transformations on the solution quality of quantum annealing: A benchmark study.
\newblock In {\em Proceedings of the Companion Conference on Genetic and Evolutionary Computation}, pages 2263--2271.

\bibitem[Zielinski et~al., 2023b]{zielinski2023pattern}
Zielinski, S., N{\"u}{\ss}lein, J., Stein, J., Gabor, T., Linnhoff-Popien, C., and Feld, S. (2023b).
\newblock Pattern qubos: Algorithmic construction of 3sat-to-qubo transformations.
\newblock {\em Electronics}, 12(16):3492.

\bibitem[Zou, 2023]{zou2023multiscale}
Zou, P. (2023).
\newblock Multiscale quantum approximate optimization algorithm.
\newblock {\em arXiv preprint arXiv:2312.06181}.

\end{thebibliography}

\end{document}